\documentclass[11pt, a4paper]{article}
\pdfoutput=1 %

\usepackage{fullpage}

\usepackage{mdframed}

\usepackage[utf8]{inputenc}
\usepackage[T1]{fontenc}

\usepackage{libertine}

\usepackage{amsfonts}
\usepackage{amsmath}
\usepackage{amssymb}
\usepackage{amsthm}
\usepackage{float}
\usepackage{stmaryrd}

\usepackage{cmap}
\usepackage{hyperref}
\usepackage[svgnames]{xcolor}
\hypersetup{colorlinks={true},urlcolor={blue},linkcolor={DarkBlue},citecolor=[named]{DarkGreen}}
\usepackage[authoryear,square]{natbib}

\usepackage{xspace}
\usepackage{fancyhdr}
\usepackage{tcolorbox}

\usepackage{microtype}
\usepackage[capitalise,nameinlink]{cleveref}
\usepackage{enumitem}

\usepackage{cleveref}

\usepackage{doi}

\usepackage{authblk}

\usepackage{booktabs} %
\usepackage[ruled,linesnumbered]{algorithm2e} %
\usepackage{setspace}

\SetAlFnt{\small}
\SetAlCapFnt{\small}
\SetAlCapNameFnt{\small}
\SetAlCapHSkip{0pt}
\IncMargin{-\parindent}

\usepackage{tikz}  %
\usetikzlibrary{arrows}
\usetikzlibrary{patterns,snakes}
\usetikzlibrary{decorations.shapes}
\tikzstyle{overbrace text style}=[font=\tiny, above, pos=.5, yshift=5pt]
\tikzstyle{overbrace style}=[decorate,decoration={brace,raise=5pt,amplitude=3pt}]
\usetikzlibrary{shapes.geometric}

\usepackage{authblk}
\usepackage{xspace}

\theoremstyle{definition}
\newtheorem{definition}{Definition}

\theoremstyle{plain}
\newtheorem{theorem}{Theorem}[section]

\newtheorem{corollary}[theorem]{Corollary}
\newtheorem{proposition}[theorem]{Proposition}

\newtheorem{observation}{Observation}
\newtheorem{claim}{Claim}
\newenvironment{claimproof}[1]{\par\noindent\emph{Proof.}\hspace{0.15cm}#1}{\hfill $\blacktriangleleft$\smallskip}

\theoremstyle{definition}

\newcommand{\ical}{\ensuremath{\mathcal{I}}\xspace}

\newcommand{\ocal}{\ensuremath{\mathcal{O}}\xspace}

\newcommand{\piswap}{\ensuremath{\pi^{i \leftrightarrow j}}\xspace}
\newcommand{\lambdaswap}{\ensuremath{\lambda^{i \leftrightarrow j}}\xspace}

\newcommand{\NP}{\textsf{NP}\xspace}
\newcommand{\NPh}{\NP-hard\xspace}
\newcommand{\NPc}{\NP-complete\xspace}
\newcommand{\FPT}{\textsf{FPT}\xspace}
\newcommand{\XP}{\textsf{XP}\xspace}
\newcommand{\W}[1][1]{\textsf{W[#1]}\xspace}
\newcommand{\Wh}[1][1]{\W[#1]-hard\xspace}

\newcommand{\Oh}[1]{{\ocal\left(#1\right)}}

\newcommand{\FDwE}[1]{\if\relax\detokenize{#1}\relax\else#1-\fi\textsc{Fair Division with Externalities}\xspace}
\newcommand{\FD}[1]{\if\relax\detokenize{#1}\relax\else{\normalfont #1-}\fi\textsc{Fair Division}\xspace}

\newcommand{\types}{\ensuremath{\tau}}
\newcommand{\items}{\ensuremath{A}}
\newcommand{\agents}{\ensuremath{N}}

\usepackage{todonotes}

\allowdisplaybreaks

\title{\bf The Complexity of Fair Division of Indivisible Items with Externalities\footnote{An extended abstract of this work has been published in the proceedings of the 38th AAAI Conference on Artificial Intelligence, AAAI '24~\citep{DeligkasEKS2024}.}}

\author[1]{Argyrios Deligkas}
\author[1]{Eduard Eiben}
\author[2]{Viktoriia Korchemna}
\author[3]{Šimon Schierreich}

\affil[1]{Royal Holloway, University of London, United Kingdom}
\affil[2]{TU Wien, Austria}
\affil[3]{Czech Technical University in Prague, Czechia}

\date{\today}

\begin{document}

\maketitle
\thispagestyle{empty}

\begin{abstract}
We study the computational complexity of fairly allocating a set of indivisible items under externalities.
In this recently-proposed setting, in addition to the utility the agent gets from their bundle, they also receive utility from items allocated to other agents.
We focus on the extended definitions of envy-freeness up to one item (EF1) and of envy-freeness up to any item (EFX), and we provide the landscape of their complexity for several different scenarios. 
We prove that it is \NP-complete to decide whether there exists an EFX allocation, even when there are only three agents, or even when there are only six different values for the items.
We complement these negative results
by showing that when both the number of agents and the number of different values for items are bounded by a parameter the problem becomes fixed-parameter tractable. 
Furthermore, we prove that
two-valued and binary-valued instances are equivalent and that EFX and EF1 allocations coincide for this class of instances.
Finally, motivated from real-life scenarios, we focus on a class of structured valuation functions, which we term agent/item-correlated. 
We prove their equivalence to the ``standard'' setting without externalities. 
Therefore, all previous results for EF1 and EFX apply immediately for these valuations.
\end{abstract}

\section{Introduction}
\label{sec:intro}
The {\em allocation} of a set of {\em indivisible} resources, e.g., objects, tasks, responsibilities, in a {\em fair} manner is a question that has received a lot of attention through history. %
In the last decades though, economists, mathematicians, and computer scientists have systematically started studying the problem with the aim of providing {\em formal fairness guarantees}~\citep{LiptonMMS04,BouveretL08,Budish11,CaragiannisKMPSW19}; for excellent recent surveys on the topic, see~\citep{AmanatidisABFLMVW2023} and~\citep{NguyenR2023}. 
However, despite the significant efforts on this quest, the nature of the problem, i.e., the indivisibility, has not allowed yet for a universally adopted solution concept.

Typically, an instance of the fair division problem consists of a set of indivisible items, and a set of agents each of whom has their own {\em valuation function}.
The task is to {\em partition} the items into {\em bundles} and allocate each bundle to an agent such that from the point of view of {\em every} agent this allocation is ``fair''.
Here, the ``fair'' part is a mathematical criterion that has to be satisfied by the valuation function of every agent.

Traditionally, in the majority of previous works, the mathematical criterion of fairness for each agent depends {\em only} on a pairwise comparison between {\em bundles}. 
Put simply, each agent cares, i.e., derives value, {\em only} about the bundle they receive, and they compare it against the bundle of any other agent.
However, in many real-life situations this assumption is not sufficient due to inherent underlying {\em externalities}.

Consider for example the scenario where there is a set of admin tasks that have to be assigned to the faculty members of a CS department.
There could exist certain tasks such that some faculty members are objectively better qualified for them -- and would even enjoy doing them -- while other faculty members are not that suited for them. Here, {\em every} faculty member evaluates the allocation as a {\em whole}, since they are affected, either positively or negatively, by the quality of completion of (almost) all tasks.

As a different example, assume that the agents are a priori partitioned into two teams, Team A and Team B, that compete against each other and consider a specific agent in Team A. 
Then, for any resource the agent considers as {\em good}, they will get positive value if it is allocated to Team A -- maybe the value is discounted compared to the value the agent would get if they got the item -- while they get zero, or even negative, value if the resource is allocated to Team B. 
At the same time, for any resource/task that the agent considers that will decrease the efficiency of the team, i.e. they view it as a {\em chore}, they would get negative value if it is allocated to some other agent from Team A.

Motivated by real-life scenarios like the two above, \citet{efx-externalities} recently proposed a new model suitable to capture the situations where externalities occur; interestingly, for \emph{divisible} items, the first models that incorporate externalities were proposed many years ago~\citep{BranzeiPZ13,LiZZ15}. 
The foundational principle of their model is that the agents have additive valuations over the items where, for every item $a$, agent $i$ derives value $V_i(j,a)$ if agent $j$ gets the item; here $j$ can be equal to~$i$. 
This way, every agent evaluates the {\em entire} allocation and not just their bundle.
Furthermore, in view of the more general valuation functions~\citeauthor{efx-externalities} appropriately extended the most prominent fairness concepts for indivisible items: \emph{envy-freenes} (EF), \emph{envy-freeness up to one item} (EF1)~\citep{LiptonMMS04,Budish11}, and \emph{envy-freeness up to any item} (EFX)~\citep{CaragiannisKMPSW19}. Intuitively, they are defined as follows.

An allocation is EF1 if for agent $i$ that prefers the allocation where they swap bundles with agent $j$, there exists one item in either of the bundles of agents $i$ and $j$ (depending on whether it is a good or a chore), such that by removing it agent $i$ does not longer prefers the allocation with the bundles swapped.
An allocation is EFX if instead of removing {\em some} item from the bundles of agents $i$ and $j$ in order to eliminate agent's $i$ preference towards the allocation with swapped bundles, it suffices to remove {\em any} item from the same bundles that {\em strictly} decreases the envy of agent $i$ towards the allocation.

In contrast to the basic setting without externalities, where EF1 allocations always exist and EFX allocations are guaranteed to exist for a few settings, \citet{efx-externalities} showed that things become significantly more complicated in the presence of externalities.
While it is currently unknown, and a major open problem, whether in the basic setting EFX allocations always exist\footnote{The problem is open only for goods or only for chores. Recently,~\citet{HosseiniSVX23-lexicographic} 
resolved the problem for mixed items.}, \citet{efx-externalities} show that there exist instances with externalities without any EFX allocation! 
However, for those instances where existence of an EFX allocation is not guaranteed, \citeauthor{efx-externalities} do not provide any results for the associated computational problem, i.e., decide whether a fair allocation exists or not. 
We resolve this open problem and we deep dive into the uncharted waters of the computational 
complexity of fair division with externalities.

\subsection{Our Contribution}  

We begin our study of the complexity of fair division with externalities by proving that it is intractable to decide whether a given instance admits an EFX allocation, even for very restricted settings.
Firstly, we show that it is \NPc to solve the problem, even when there are only three agents.
This paints a clear dichotomy between tractable and intractable cases, as \citet{efx-externalities} showed that for two agents an EFX allocation can always be found in polynomial time. This result also shows that fair division with externalities is significantly harder compared to the standard setting without externalities, where polynomial algorithm is known for two agents~\citep{GoldbergHH2023}, and pseudo-polynomial algorithms exist for instances with three~\citep{ChaudhuryGM2020} and partly with four agents~\citep{BergerCFF2022,GhosalPNV2023}.

Next, we restrict the problem at a different dimension and we turn our attention to instances with valuations that use only a small number of values. We prove that the problem remains \NP-complete even if the domain of the valuation function consists of $6$ different values.
It is also worth mentioning that in our hardness constructions we do not exploit the presence of chores, as is common in standard fair division settings~\citep{HosseiniMW2023}.

In light of our hardness lower bounds, we use the framework of parameterized complexity~\citep{Niedermeier2006,DowneyF2013,CyganFKLMPPS2015} to reveal at least some tractable fragments of the problem. It should be pointed out that this framework has become \emph{de facto} standard approach when dealing with \NPh problems in AI, ML, and computer science in general~\citep{KroneggerLPP2014,IgarashiBE2017,BredereckKKN2019,GanianK2021,DeligkasEGHO2021,BlazejGKPSS2023}. 
Roughly speaking, in this framework, we study the %
complexity of a problem not merely with respect to the input size~$n$, but also assuming additional information about the instance captured in the so-called \emph{parameter}\footnote{We provide a formal introduction to parameterized complexity in Preliminaries.}.

We start our algorithmic journey with the combined parameter: the number of {\em item types} and the number of agents. 
Intuitively, two items are of the same type if all agents value them the same if they are allocated to a distinct agent $j$; this parameter was recently initiated by \citet{GorantlaMV2023} for goods and by \citet{AzizLRS2023} for chores. 
As our results indicate, the combination of these two parameters is necessary in order to achieve fixed parameter tractability; our first negative result holds just for three agents, and our second result produces an instance with just three different item types.
Hence, our algorithm is the best possible one could hope for, and actually it is capable of finding also EF and EF1 allocations, if they exist.

Moreover, our algorithm serves as the foundation for our second positive result, which is an efficient procedure deciding the existence of an EFX/EF1/EF allocation for the combined parameter the number of agents and the number of different values in agents' preferences. The latter parameter naturally captures widely studied binary valuations~\citep{BarmanKV2018,FreemanSVX2019,HalpernPPS2020,BabaioffEF2021,SuksompongT2022}, bi-valued valuations~\citep{EbadianPS2022,GargMQ2022}, and was previously used by~\citet{AmanatidisBFHV2021} and \citet{GargM2023}.

Next, we move to instances with structured valuations. Following the approach of \citet{efx-externalities}, we start with binary valuations, i.e., they have $\{0,1\}$ as domain. First, we show that instances where every agent uses only two different values (which can be different for every agent) are, in fact, equivalent to binary valuations. 
Additionally, and more importantly, we show that for binary valuations EFX and EF1 allocations coincide. Thus, using the existential result of \citeauthor{efx-externalities}, we establish the existence of EFX allocations for three agents with binary valuations and no chores.

Finally, we introduce and study a different class of structured valuations which we term {\em agent/item-correlated} valuations.
Intuitively, under agent-correlated valuations, an agent $i\in\agents$ receives utility $v_{i,a}$ if an item $a$ is given to her, and~$\tau_{i,j}$ fraction of $v_{i,a}$ if the item is allocated to agent $j$. 
Item-correlated valuations are similar; however, the fractional coefficient depends on the item which is allocated not the agent who gets it. We show that instances with agent/item-correlated valuations can be turned into equivalent instances of fair division without externalities with the same sets of agents and items.  To conclude, we show how the agent- and item-correlated preferences capture many real-life scenarios such as team preferences~\citep{IgarashiKSS2023}.
 
\section{Preliminaries}
\label{sec:prelims}

We will follow the model and the notation of~\citet{efx-externalities}. There is a set of indivisible items $\items=\{a_1, a_2, \ldots, a_m\}$ and a set of agents $N=\{1, 2, \ldots, n\}$. 
An allocation $\pi = (\pi_1, \pi_2, \ldots, \pi_n)$ is a partition of the items into $n$ possibly empty sets, i.e., $\pi_i \cap \pi_j=\emptyset$ for every $i \neq j$ and $\bigcup_{i\in N}\pi_i = \items$, where set $\pi_i$ is allocated to agent $i$. 
Let $\Pi$ denote the set of all allocations. 
For any item $a \in \items$, denote $\pi(a)$ the agent who receives item $a$ in allocation $\pi$.

We assume that the agents have valuation functions with {\em additive externalities}. More formally, every agent $i$ has a value $V_i(j,a)$ for every item $a \in \items$ and every agent $j \in N$; put simply, agent $i$ gets value $V_i(j,a)$ if item $a$ is allocated to agent $j$.
The value of agent $i$ from allocation $\pi$ is $V_i(\pi) = \sum_{a \in \items} V_i(\pi(a),a)$.

Let $a\in \items$ be an item. We define the {\em item-type} as a vector $\left(V_1(1,a),\ldots,V_1(|\agents|,a),V_2(1,a),\ldots,V_{|\agents|}(|\agents|,a)\right)$. 
That is, two items are of the same item-type if the associated vectors are the same; intuitively, these items are ``indistinguishable'' from the point of view of every agent. 
By $\types$ we denote the number of different item-types in an instance.

We will focus on {\em envy-freeness} and its relaxations, EF1 and EFX, in the presence of externalities. 
Since now every agent evaluates the {\em whole allocation} and not just their bundle like in the no-externalities case, the idea of {\em swapping bundles} needs to be deployed. 
We use $\piswap$ to denote the new allocation in which agents $i$ and $j$ swap their bundles in~$\pi$ while the bundles of the other agents remaining the same.

\begin{definition}[EF~\citep{velez2016fairness}]
An allocation $\pi$ is envy-free (EF), if for every pair of agents $i,j \in N$ it holds that $V_i(\pi) \geq V_i(\piswap)$.
\end{definition}

\begin{definition}[EF1~\citep{efx-externalities}]
An allocation $\pi$ is envy-free up to one item (EF1), if for every pair of agents $i,j \in N$ there exists an item $a \in \items$ and an allocation $\lambda$ such that:
(i) $\lambda_\ell = \pi_\ell \setminus \{a\}$, for all $\ell \in N$; and
(ii) $V_i(\lambda) \geq V_i(\lambdaswap)$.
\end{definition}

\begin{definition}[EFX~\citep{efx-externalities}]\label{def:EFX}
An allocation $\pi$ is envy-free up to any item (EFX), if for every pair of agents $i,j \in N$, if $V_i(\pi) < V_i(\piswap)$, then for any item $a \in \items$ and allocation $\lambda$ with the properties
\begin{enumerate}
    \item $\lambda_\ell = \pi_\ell \setminus \{a\}$, for all $\ell \in N$;
    \item $V_i(\lambda) - V_i(\lambdaswap) > V_i(\pi) - V_i(\piswap)$,
\end{enumerate}
it holds that $V_i(\lambda) \geq V_i(\lambdaswap)$.
\end{definition}

Observe that the second property of Definition~\ref{def:EFX} above implies that we {\em have to} remove an item from $\pi_i \cup \pi_j$ with {\em strictly} non-zero value for agent $i$; depending from which bundle we remove the item, it can be either a good, or a chore.
This definition is equivalent to EFX in the absence of externalities, when we have to remove items of non-zero value from a bundle an agent envies. A different, more constrained version, termed EFX$_{0}$, requires that envy should be eliminated by removing any item, even if the agent has zero value for it~\citep{PlautR2020}.

For allocation $\pi$ and two agents $i,j\in N$, we say that~$i$ \emph{envies} $j$ or alternatively that there is \emph{envy from $i$ towards~$j$} if $V_i(\piswap) > V_i(\pi)$. We will use the following simple observation in our hardness proofs.

\begin{observation}
    If %
    an agent $i$ does not envy agent $j$, then the pair $i,j$ satisfies Definition~\ref{def:EFX} for every item $a\in \items$.
\end{observation}
\begin{proof}
    Let $a$ be an item in $\items$ and $\lambda$ be the allocation as is Property 1 of Definition~\ref{def:EFX}. If Property 2 is not satisfied, then we are done. Hence, by Property 2  $V_i(\lambda) - V_i(\lambdaswap) > V_i(\pi) - V_i(\piswap)$. Moreover, since $i$ does not envy $j$, we have  $V_i(\pi) \ge V_i(\piswap)$, so $V_i(\lambda) - V_i(\lambdaswap) > V_i(\pi) - V_i(\piswap) \ge 0$ and $V_i(\lambda) > V_i(\lambdaswap)$.
\end{proof}

Finally, we are ready to define the computational problems we will study.

\begin{definition}
\label{def:FDwE}
Let $\phi \in \{\text{EF, EF1, EFX}\}$. An instance $\ical=(N,\items,V)$ of $\FDwE{\phi}$ consists of a set of items $\items$ and a set of agents $\agents$ with valuation functions $V$ with additive externalities.
The task is to decide whether there exists an allocation that is fair with respect to solution concept~$\phi$.
\end{definition}

\paragraph{Normalized valuations.}%
Valuations with externalities allow us to consider only non-negative values in the valuations, i.e. we can ``normalize'' them as follows.
Let $i\in\agents$ be an agent. For each item $a \in \items$ we compute $x_{i,a} := \min_{j \in \agents} V_i(j,a)$ and we set $V_i(j,a) \leftarrow V_i(j,a) - x_{i,a}$.

\begin{proposition}\label{prop:normalized-valuations}
    Let $\ical$ be an instance of $\FDwE{\phi}$. Then, we can get an instance $\ical'$ with normalized valuations such that any solution for instance $\ical'$ corresponds to a solution for~$\ical$.
\end{proposition}
\begin{proof}
    Let $\ical'$ be the normalized instance obtained as above, i.e., for each item $a \in \items$ and each agent $i\in \agents$ we compute $x_{ia} = \min_{j \in \agents} V_i(j,a)$ and we set $V'_i(j,a) = V_i(j,a) - x_{ia}$.

    Let $\pi$ be any allocation. Then 
    \begin{align*}
     V'_i(\pi) &= \sum_{a \in \items} V'_i(\pi(a),a) 
     = \sum_{a \in \items} (V_i(\pi(a),a) - x_{ia})\\ 
     &= V_i(\pi) - \sum_{a \in \items}x_{ia}.
    \end{align*}
    It follows that for any $j\in \agents$ we have 
    \begin{align*}
        V'_i(\pi) &-V'_i(\piswap)  \\ & = (V_i(\pi)  - \sum_{a \in \items}x_{ia}) - (V_i(\piswap) - \sum_{a \in \items}x_{ia}) \\ &= V_i(\pi) - V_i(\piswap).
    \end{align*}

    Therefore, it is easy to observe that $\pi$ is EF for $\ical$ if and only if it is EF for $\ical'$.

    Similarly, for any item $a'$, if $\lambda$ is obtained from $\pi$ such that $\lambda_i = \pi_i\setminus\{a'\}$ for all $i\in \agents$, then we get $V'_i(\lambda) = V_i(\pi) - \sum_{a \in \items\setminus\{a'\}}x_{ia}$ and $ V'_i(\lambda) -V'_i(\lambdaswap) = V_i(\lambda) - V_i(\lambdaswap)$. Therefore, it follows rather straightforwardly that $\pi$ is $\phi$ for $\phi \in \{\text{EF1, EFX}\}$ for $\ical$ if and only if it is $\phi$ for $\ical'$.
\end{proof}

\paragraph{Chores.}
While in the standard setting the definition of chores is straightforward, in the presence of externalities they can be defined in more than one way; \citet{efx-externalities} defined them informally. Below, we define {\em strong-chores} and {\em weak-chores}. Intuitively, a strong-chore is an item that an agent does not want to have at all; this resembles the ``standard'' chore-definition. On the other hand, an item is a weak-chore, if the agent does not mind having it, but there exist some other agents that it would be better for him if they get it; so, weak-chores capture positive externalities. A strong-chore is a weak-chore, but not vice versa. Hence, negative results with respect to weak-chores carry over to strong-chores.

\begin{definition}[Strong/weak-chore]\label{def:strong-weak-chores}
    An item $a$ is a {\em strong-chore} for agent $i$ if $V_i(i,a) \leq V_i(j,a)$ {\em for all} $j \neq i$, where for at least one $j$ the inequality is strict. 
    An item $a$ is a {\em weak-chore} for agent $i$ if there exists an agent $j$ such that $V_i(i,a) < V_i(j,a)$.
\end{definition}

\paragraph{Parameterized Complexity.} 
An instance of a parameterized problem $Q\subseteq \Sigma\times\mathbb{N}$, where $\Sigma$ is fixed and finite alphabet, is a pair $(I,k)$, where $I$ is an input of the problem and~$k$ is \emph{parameter}. The ultimate goal of parameterized algorithmics is to confine the exponential explosion in the running time of an algorithm for some \NPh problem to the parameter and not to the instance size. In this line of research, the best possible outcome is the so-called \emph{fixed-parameter} algorithm with running time $f(k)\cdot |I|^\Oh{1}$ for any computable function $f$. 
That is, for every fixed value of the parameter, we have a polynomial time algorithm where, moreover, the degree of the polynomial is independent of the parameter. 
Based on this, we define the complexity class~\FPT, which contains all \emph{fixed-parameter tractable} parameterized problems -- problems that admit a fixed-parameter algorithm.

Naturally, not all parameterized problems admit a fixed-parameter algorithm. The strongest intractability result is to show that some parameterized problem is \NPh even for a constant value of the parameter -- we call such parameterized problems para-\NPh. To exemplify this, assume parameterized problem \textsc{Vertex Coloring} with the number of colors as the parameter. It is well-known that this problem is \NPh already for $k=3$. Hence, \textsc{Vertex Coloring} parameterized by the number of colors is clearly not in \FPT (or \XP), as the existence of a fixed-parameter tractable algorithm would yield $\mathsf{P}=\NP$, which is unlikely.

For a more comprehensive introduction to parameterized complexity, we refer the interested reader to the monograph of~\citet{CyganFKLMPPS2015}.

\section{General Valuations}
\label{sec:general}
In this section, we focus on the case of general valuations for the agents. Our first negative result shows that it is intractable to decide whether there is an EFX allocation even when there are three agents and there are no chores.

\begin{theorem}\label{thm:NP-h:three-agents}
    \FDwE{EFX} is \NP-complete, even if there are three agents and there are no weak-chores.
\end{theorem}
\begin{proof}
    Firstly, it is not hard to see that we can verify in polynomial time whether an allocation is EFX, since for any agent we can simply calculate whether they envy an allocation where we swap two bundles and whether this can be eliminated by removing each item that satisfies the properties form Definition~\ref{def:EFX}.

    Next, we prove hardness by providing a polynomial reduction from the \textsc{Equal-Cardinality-Partition} problem in which, given a sequence $S=(s_1,s_2,\ldots,s_{2n})$ of $2n$ integers for some $n\in \mathbb{N}$, the goal is to find 
    a subset $I\subseteq [2n]$ of size $n$ such that $\sum_{i\in I}s_i = \sum_{i\in [2n]\setminus I} s_i$. 
    This version of NP-hard \textsc{Partition} problem~\citep{GareyJohnson79} is easily seen to be also NP-hard, as we can just add $|S|$ many copies of integer $0$ to achieve balanced sizes of the subsets. 
    Observe that the problem remains NP-hard even if we assume that $n>10$, since for $n\le 10$, we can just try all possibilities in polynomial time. 
    Let $s_{\min}$ and $s_{\max}$ be the minimum and the maximum integers in $S$ respectively, and let $M=(s_{\max}-s_{\min})\cdot n^2$. We can assume that $M> 0$, otherwise all the numbers in $S$ are equal and hence any subset of size $n$ forms a solution. 
    
    We construct a new instance $S'=(s'_1, s'_2,\ldots, s'_{2n})$ of \textsc{Equal-Cardinality-Partition} by shifting all the numbers in $S$ by a constant: $s'_i = M+ s_i - s_{\min}$. Then for any $I\subseteq [2n]$ of size $n$ it holds that $\sum_{i\in I}s'_i = \sum_{i\in [2n]\setminus I}s'_i$ if and only if $\sum_{i\in I}s_i = \sum_{i\in [2n]\setminus I}s_i$. 
     Let us now denote $B=\frac{1}{2}\sum_{i\in [2n]}(s_i-s_{\min})$, then  \textsc{Equal-Cardinality-Partition} asks to find $I\subseteq [2n]$ of size $n$ such that $$\sum_{i\in I}s'_i=\frac{1}{2}\sum_{i\in [2n]}s'_i=\frac{1}{2}\sum_{i\in [2n]}(M+ s_i - s_{\min})=Mn+B.$$ In addition, note that 
     \begin{align*}
       B=&\frac{1}{2}\sum_{i\in [2n]}(s_i-s_{\min})
          \le \frac{1}{2}\sum_{i\in [2n]}(s_{\max}-s_{\min})= \\
          &= n\cdot (s_{\max}-s_{\min}) < n^2\cdot (s_{\max}-s_{\min}) = M,
     \end{align*}
    so $B < M$. By construction, $M \le s'_i \le M+s_{\max} - s_{\min}$ for every $i\in [2n]$. Therefore, for any set $I\subseteq [2n]$ we have $M\cdot |I| \le \sum_{i\in I}s'_i \le |I|\cdot (M + s_{\max}-s_{\min}) < M\cdot (|I|+1)$.
    
     We now create an equivalent instance of \FDwE{EFX} with three agents and no chores, where the set $\items=\{a_i\:|\:i\in [2n+2]\}$ consists of $2n+2$ items, first $2n$ associated with integers in $S$ and two auxiliary items. The valuations are defined as follows:
     \begin{itemize}
         \item $V_i(i, a_j) = s'_j$ for $i\in [3]$ and $j\in [2n]$;
         \item $V_1(1, a_{2n+1})=V_2(2, a_{2n+2})=Mn+B$;
         \item $V_1(1, a_{2n+2})=V_2(2, a_{2n+1})=1$;
         \item $V_1(2, a_{2n+2})=V_2(1, a_{2n+1})=-M^2$;
         \item $V_3(3, a_{2n+1})=V_3(3, a_{2n+2})=\frac{Mn+B}{2}$;
         \item all remaining values are zeros.
     \end{itemize}
    
    We start by showing how to obtain an EFX allocation $\pi$ from any $I\subseteq [2n]$ such that $\sum_{i\in I}s'_i = \sum_{i\in [2n]\setminus I}s'_i$. We set $\pi = (\{a_i\mid i\in I\}, \{a_i\mid i\in [2n]\setminus I\}, \{a_{2n+1}, a_{2n+2}\})$. 
    
    First, observe that $V_1(\pi) = V_2(\pi) = V_1(\pi^{1\leftrightarrow 2}) = V_2(\pi^{1\leftrightarrow 2}) = V_3(\pi) = V_3(\pi^{1\leftrightarrow 3}) = V_3(\pi^{2\leftrightarrow 3}) = Mn+B$ and $ V_1(\pi^{1\leftrightarrow 3}) = V_2(\pi^{2\leftrightarrow 3}) = Mn+B+1$. Therefore, the only envy is from the agents 1 and 2 towards the agent 3. Since there are no chores, removing items from $\pi_1$ or $\pi_2$ does not decrease the envy. Moreover, removing $a_{2n+1}$ from $\pi_3$ eliminates the envy. Indeed, if $\lambda$ is such that $\lambda_{\ell} = \pi_{\ell}\setminus a_{2n+1}$, then $V_1(\lambda) = V_2(\lambda) = Mn+B$, while  $V_1(\lambda^{1\leftrightarrow 3}) = 1$ and  $V_2(\lambda^{2\leftrightarrow 3}) = Mn+B$.  Similarly, removing $a_{2n+2}$ from $\pi_3$ eliminates the envy. Hence, $\pi$ is EFX.
    
    On the other hand, assume that $\pi=(\pi_1, \pi_2, \pi_3)$ is EFX. For every $j\in[3]$, let $I_j = \{i\mid a_{i}\in \pi_j\}\cap [2n]$. We will show that $I_1$ is a solution to \textsc{Equal-Cardinality-Partition} instance $S$. 
    For this, we will distinguish all possibilities depending on where $a_{2n+1}$ and $a_{2n+2}$ belong and show that the only possible case is $\{a_{2n+1},a_{2n+2}\}\subseteq \pi_3$, since in rest of the cases, the assignment $\pi$ cannot be EFX.

    According to Definition \ref{def:EFX}, the fact that $\pi$ is not EFX can be witnessed by a pair of agents $(i,j)$ and an item $a\in \pi_i \cup \pi_j$ such that $V_i(\piswap) - V_i(\pi) > V_i(\lambdaswap) - V_i(\lambda) >0$, where $\lambda$ is the allocation obtained from $\pi$ by removing $a$. In our examples, we will always remove $a$ from the bundle of $j$. Then, since for any item $a$ it holds that $V_i(i,a)>V_i(j,a)$ whenever $i\ne j$, the first inequality $V_i(\piswap) - V_i(\pi) > V_i(\lambdaswap) - V_i(\lambda)$ will be satisfied automaticaly. Therefore, it will suffice to ensure that $V_i(\lambdaswap) - V_i(\lambda) >0$.
    
     To simplify the case distinction, we observe the following:
    \begin{observation}
    \label{obs: i1_empty}
    If $a_{2n+1} \in \pi_1$, then $I_1 = \emptyset$.
    \end{observation}
    \begin{proof}
    Assume that $I_1\neq \emptyset$, and let $\lambda$ be the 
    assignment obtained from $\pi$ by removing $a_{i_1}$  from $\pi_1$ for some $i_1\in I_1$. 
    Then we have $V_2(\lambda) \le -M^2 + Mn+B +\sum_{i\in I_2}s'_i\le -M^2 +3Mn+3B = -Mn^2(s_{\max}-s_{min}) + (3n+3)M \le$ \\ $\le(3n+3-n^2)M< 0$ for any $n>10$. But this is strictly smaller than $V_2(\lambda^{1\leftrightarrow 2}) \ge 1+ \sum_{i\in I_1\setminus\{i_1\}}s'_i > 0$, which contradicts to the fact that $\pi$ is EFX.
    \end{proof}
    Furthermore, note that swapping the agents 1 and 2 and at the same time items $a_{2n+1}$ and $a_{2n+2}$ yields the same instance. This symmetry allows to easily handle some cases. 
    
    \textbf{Case 1. $\{a_{2n+1},a_{2n+2}\}\subseteq \pi_3$}. Let $\lambda$ be obtained by removing $a_{2n+2}$ from $\pi_3$, then $V_1(\lambda^{1\leftrightarrow 3}) \ge Mn+B$. As $\pi$ is EFX, we have $\sum_{i\in I_1}s'_i=V_1(\pi)=V_1(\lambda) \ge Mn+B$. Similarly, $\sum_{i\in I_2}s'_i\ge Mn+B$. The only way to achieve this by splitting items $a_i$, $i\in [2n]$, between the agents is to split them such that $\sum_{i\in I_1}s'_i = \sum_{i\in I_2}s'_i = Mn+B$. In particular, $I_1$ is a solution to $S$.
    
    \textbf{Case 2. $a_{2n+1}\in \pi_1$, $a_{2n+2}\in \pi_3$.} 
    By Observation \ref{obs: i1_empty} we have $I_1=\emptyset$, so $V_1(\pi) = Mn+B$ and $V_1(\pi^{1\leftrightarrow 2}) = \sum_{i\in I_2}s'_i$. If $|I_2|\ge n+2$, then for $\lambda$ obtained from $\pi$ by removing $a_{i_2}$ for some $i_2 \in I_2$, we have $V_1(\lambda^{1\leftrightarrow 2}) = \sum_{i\in I_2\setminus \{i_2\}}s'_i\ge |I_2\setminus\{i_2\}|\cdot M\ge (n+1)\cdot M > Mn+B = V_1(\lambda)$, which is impossible for an EFX allocation. Hence $|I_2|\le n+1$, which along with $|I_1|+|I_2|+|I_3|=2n$ yields that $|I_3|\ge n-1$. 
    
    Let $\lambda$ be obtained from $\pi$ by removing arbitrary $a_3 \in \pi_3$, then $V_2(\lambda) = -M^2+ \sum_{i\in I_2}s'_{i} < -M^2 +M(|I_2|+1) \le -M^2+ M(n+2)$. On the other hand, since $a_{n+2}\in \pi_3$, we have $V_2(\lambda^{2\leftrightarrow 3})=-M^2+Mn+B+\sum_{i\in I_3\setminus\{i_3\}}s'_i \ge -M^2+Mn+B + (n-2)M > -M^2+ M(2n-2)$. Since $n>10$, we have $V_2(\lambda^{2\leftrightarrow 3}) > V_2(\lambda)$, which contradicts to $\pi$ being EFX. 
    
    \textbf{Case 3. $a_{2n+2}\in \pi_2$, $a_{2n+1}\in \pi_3$}: analogous to \textbf{Case 2}.
    
    \textbf{Case 4. $a_{2n+2}\in \pi_1$, $a_{2n+1}\in \pi_3$.} If $I_3\neq \emptyset$, let $\lambda$ be the assignment obtained from $\pi$ by removing $a_{i_3}$ for some $i_3\in I_3$. Then $V_1(\lambda) = 1 + \sum_{i\in I_1}s'_i \le (|I_1|+1)M$ and $V_1(\lambda^{1\leftrightarrow 3}) = Mn+B + \sum_{i\in I_3\setminus\{i_3\}}s'_i > Mn$. Since $\pi$ is EFX, we have $V_1(\lambda)\ge V_1(\lambda^{1\leftrightarrow 3})$ and hence $|I_1|\ge n$. In particular, $|I_2|\le n-1$.
    
    Let $\lambda$ be obtained from $\pi$ by removing $a_{i_1}$ for some $i_1\in I_1$, then $V_2(\lambda) = \sum_{i\in I_2}s'_i \le (|I_2|+1)M\le Mn$ and $V_2(\lambda^{1\leftrightarrow 2}) = Mn+B+\sum_{i\in I_1\setminus \{i_1\}}s_i'\ge Mn+B$, so $V_2(\lambda^{1\leftrightarrow 2}) > V_2(\lambda)$, which contradicts to $\pi$ being EFX. 
    
    Hence, $I_3=\emptyset$. Then either $|I_2|\ge n$ or $|I_1|\ge n$. Let's first consider the case when $|I_2|\ge n$. For any $i_2\in I_2$, let $\lambda$ be the assignment obtained from $\pi$ by removing $a_{i_2}$. We get $V_3(\lambda) = \frac{Mn+B}{2} < M\frac{n+1}{2}$ and $V_3(\lambda^{2\leftrightarrow 3}) = \sum_{i\in I_2\setminus\{i_2\}}s_i' \ge M(n-1)$ so $V_3(\lambda^{2\leftrightarrow 3}) > V_3(\lambda)$ and the assignment is not EFX. The case $|I_1|\ge n$ is even more straightforward, as for any $i_1\in I_1$ and $\lambda$ obtained by removing $a_{i_1}$ from $\pi$, we have $V_3(\lambda) = \frac{Mn+B}{2}$ and $V_3(\lambda^{1\leftrightarrow 3}) = \frac{Mn+B}{2} + \sum_{i\in I_1\setminus\{i_1\}}s_i'$ so $V_3(\lambda^{1\leftrightarrow 3}) > V_3(\lambda)$. Therefore, it is not possible to achieve this case.

    \textbf{Case 5. $a_{2n+1}\in \pi_2$, $a_{2n+2}\in \pi_3$}: analogous to \textbf{Case 4}.
    
    \textbf{Case 6. $a_{2n+1}\in \pi_1$, $a_{2n+2}\in \pi_2$.} 
    By Observation \ref{obs: i1_empty} we have $I_1=\emptyset$. Symmetrically, $I_2=\emptyset$ and hence $I_3=[2n]$. Then for any item $i_3\in I_3$, removing $a_{i_{3}}$ decreases the envy from agent 1 towards agent 3. Moreover, if $\lambda$ is obtained from $\pi$ by removing $a_{i_3}$, then $V_1(\lambda) = -M^2+Mn+B$ but $V_1(\lambda^{1\leftrightarrow 3}) = - M^2 + \sum_{i\in I_3\setminus\{i_3\}}s'_i \ge -M^2 + M(2n-1) >  -M^2 +  M(2n-2) + B > -M^2+Mn+B = V_1(\lambda)$, so this case is also not possible. 
    
    \textbf{Case 7. $a_{2n+2}\in \pi_1$, $a_{2n+1}\in \pi_2$.} First, assume that $|I_1|+2 \le |I_3|$, then for $\lambda$ obtained from $\pi$ by removing $a_{i_3}$ for some $i_3\in I_{3}$ we have $V_1(\lambda) = 1 + \sum_{i\in I_1}s'_i\le 1+ |I_1|\cdot (M + s_{\max}-s_{\min}) < (|I_1|+1)M$ and $V_1(\lambda^{1\leftrightarrow 3}) = \sum_{i\in I_3\setminus \{i_3\}}s'_i\ge (|I_3|-1)M\ge (|I_1|+1)M$, so $V_1(\lambda^{1\leftrightarrow 3}) > V_1(\lambda)$. It follows that $|I_1|\ge |I_3|-1$. Using a similar argument, we get that $|I_2|\ge |I_3|-1$. 
    
    If now $I_1\neq \emptyset$, then by removing any $a_{i_1}$ for $i_1\in I_1$ from $\pi$ to obtain $\lambda$, we get $V_3(\lambda) = \sum_{i\in I_3}s'_i < M(|I_3|+1)$ and $V_3(\lambda^{1\leftrightarrow 3})= \frac{Mn+B}{2} + \sum_{i\in I_1\setminus \{i_1\}} s'_i\ge M(\frac{n}{2}+|I_1|-1)$. Since $|I_1|\ge |I_3|-1$, we get $V_3(\lambda^{1\leftrightarrow 3}) > V_3(\lambda)$. Hence $I_1=\emptyset$. An analogous argument implies $I_2=\emptyset$, and so $I_3 = [2n]$. But this contradicts to $|I_1|\ge |I_3|-1$.
    
    \textbf{Case 8. $a_{2n+1}, a_{2n+2}\in \pi_1$.} Let $\lambda$ be the assignment we obtain from $\pi$ after removing $a_{2n+2}$ from $\pi_1$. Then $V_2(\lambda) = -M^2 + \sum_{i\in I_2}s'_i\le -M^2 + (2n+1)M < 0$ and $V_2(\lambda^{1\leftrightarrow 2}) = 1+ \sum_{i\in I_1}s'_i > 0$, so $V_2(\lambda^{1\leftrightarrow 2}) > V_2(\lambda)$. Therefore $\pi$ is not EFX and this case is not possible.
    
    \textbf{Case 9. $a_{2n+1}, a_{2n+2}\in \pi_2$}: analogous to \textbf{Case 8}.
    
    This finishes all possible cases. We saw that the only possibility for $\pi$ to be EFX allocation is when $\{a_{2n+1},a_{2n+2}\}\subseteq \pi_3$ in which case $I_1$ is a solution to \textsc{Equal-Cardinality-Partition} instance~$S$. 
\end{proof}

Next we restrict the problem at a different dimension and we constrain the valuation function. As our next theorem shows, the problem remains hard even if we severely limit the different values in the valuation functions of the agents and the number of different item types.

\begin{theorem}\label{thm:EFX:NPh:d}
    \FDwE{EFX} is \NP-complete even if the valuation function uses only $6$ different values, $3$ item types, and there are no weak-chores.
\end{theorem}
\begin{proof}
    We prove the statement by providing a polynomial reduction from \textsc{Min Bisection} problem on cubic graphs~\citep{bui1987graph}, i.e., the graphs where every vertex has degree precisely three. In \textsc{Min Bisection} we are given a graph $G=(V,E)$ on $2n$ vertices and an integer $k$ and the question is whether there exists a partition $(X,Y)$ of $V$ such that $|X|=|Y|=n$ and there are at most $k$ edges in $E$ with one endpoint in $X$ and the other endpoint in $Y$. 
    We can assume that no partition of $V$ into two equal parts has only at most $k-1$ edges across. 
    Finally, we will assume that $n>k> 4$, otherwise we can solve the problem in polynomial time by trying all subsets of $k$ edges and checking whether they form a cut between two parts of size $n$. 
    
    Note that if $G=(V,E)$ is a cubic graph with $2n$ vertices and $(X,Y)$ is a partition of $V$ into two subsets of size $n$ such that there are precisely $k$ edges with one endpoint in $X$ and the other in $Y$, then there are $\frac{3n-k}{2}$ edges with both endpoints in $X$ and $\frac{3n-k}{2}$ edges with both endpoints in $Y$. We are now ready to describe our reduction to \FDwE{EFX}.

    \newcommand{\xii}{\ensuremath{10n^2}}
    \newcommand{\dij}{4n}
    \newcommand{\xij}{\ensuremath{10n^2-\dij}}
    \newcommand{\yii}{\ensuremath{5n^2}}
    \newcommand{\yij}{\ensuremath{5n^2-\dij}}

    \begin{itemize}
        \item The set of agents is $N=\{1,\ldots, 3n\}$, where each agent $i\in [3n]$ is associated with an edge $e_i\in E$; 
        \item The set of items $\items$ is split into three sets (item types):
        \begin{itemize}
            \item set $X$ of $\frac{3n+k}{2}$ items; 
            \item set $Y$ of $\frac{3n-k}{2}$ items;
            \item set $Z$ of $3n-k$ items.
        \end{itemize}
        \end{itemize}
        Intuitively, we want to link any potential solution $(X',Y')$ of \textsc{Min Bisection} to the following allocation $\pi$ of items. Agent $i$, associated with edge $e_i$, receives some item from~$Y$ only if $e_i$ has both endpoints in $Y'$. Otherwise, $i$ receives some item from $X$. In addition, $i$ receives an item from $Z$ if and only if $e_i$ does not belong to the cut $(X',Y')$. To make such an allocation EFX, we define the valuations as follows.
        \begin{itemize}
        \item For all items $x\in X$, all $i\in [3n]$, $V_i(i, x) = \xii$.
        \item For all items $x\in X$, all $i, j\in N$ such that $e_i$ and $e_j$ do not share an endpoint $V_i(j,x) = \xij$.
        \item For all items $y\in Y$, all $i\in [3n]$, $V_i(i, y) = \yii$.
        \item For all items $y\in Y$, all $i, j\in N$ such that $e_i$ and $e_j$ do not share an endpoint $V_i(j,y) = \yij$.
        \item For all items $z\in Z$ and all $i\in [3n]$ we have $V_i(i, z) = 1$.
        \item All the remaining values are zero. 
    \end{itemize}
    Note that the valuation function only uses six values $0, 1, \yij, \yii, \xij, \xii$ and $3$ different item-types.

    First, let's assume $(X', Y')$ is a partition of $V$ into equal size parts such that there are exactly $k$ edges with one endpoint in $X'$ and one endpoint in $Y'$. Let $\pi$ be the allocation described above. 
    Namely: \begin{itemize}
        \item $\pi_i =\{x, z\}$ for some $x\in X$ and $z\in Z$ if both endpoints of $e_i$ are in $X'$;
        \item $\pi_i =\{y, z\}$ for some $y\in Y$ and $z\in Z$ if both endpoints of $e_i$ are in $Y'$;
        \item $\pi_i=\{x\}$ for some $x\in X$ if $e_i$ has one endpoint in $X'$ and the other endpoint is in $Y'$.
    \end{itemize}
    To see that $\pi$ is EFX, we first consider the case when edges $e_i, e_j\in E$ do not share an endpoint. If in addition the agents $i$ and $j$ receive equal numbers of items from $Z$, we have $V_i(\pi)=V_i(\piswap)$. Indeed, $i$ values items from $X$ and $Y$ by exactly $\dij$ more on itself than on $j$, and since both $i$ and $j$ have exactly one such item, the value does not change after swap. In particular, there is no envy already. 
    
    If on the other hand one of the agents, say $i$, is not assigned any item from $Z$, we have $\pi_i = \{x\}$ and $\pi_j$ contains one item from $Z$ and one item from $X\cup Y$. Then $j$ does not envy $i$ since  $V_j(\piswap)=V_j(\pi)-1$. However, $i$ envies $j$ as $V_i(\piswap)=V_i(\pi)+1$. Assume that $\lambda$ is obtained from~$\pi$ by removing one item and it decreases the envy (i.e.,~$\lambda$ satisfies the Properties 1 and 2 of Definition~\ref{def:EFX}). Since the valuations are integer-valued and $V_i(\piswap)-V_i(\pi)=1$, we have $V_i(\lambda)-V_i(\lambda)\le 0$. Hence, the envy is eliminated.
    
    It remains to consider the case when the edges $e_i$ and $e_j$ share an endpoint. Again, there are two possibilities. If $\pi_i$ and $\pi_j$ contain equal numbers of items from $Z$, then since~$e_i$ and $e_j$ share a vertex, it is straightforward to see that the unique item in $\pi_i\setminus Z$ and the unique item in $\pi_i\setminus Z$ are either both from $X$ or both from $Y$. Hence $V_i(\pi)=V_i(\piswap) = V_j(\pi)=V_j(\piswap)$ and there is no envy. 
    
    If only $\pi_j$ contains an item from $Z$, then $\pi_i =\{x\}$ for some item $x \in X$. Since all the items are goods, removing~$x$ from $\pi$ makes $\pi_i$ empty and clearly removes envy from $j$ towards $i$ if there was any envy to begin with. For agent $i$, if $\pi_j=\{y,z\}$ for some $y\in Y$ and $z\in Z$, then $V_i(\piswap) = V_i(\pi) -\xii+\yii+1 < 0$, so there is no envy from $i$ towards $j$. Finally, if $\pi_j=\{x',z\}$ for some $x'\in X$ and $z\in Z$, then $V_i(\piswap)-V_i(\pi)=1$. Hence, if removal of some item decreases the envy, it completely eliminates it.
    
    We showed that Definition~\ref{def:EFX} holds for all pairs $i,j\in N$. Therefore, if $(X', Y')$ is a partition of $G$ into two equal size parts such that there are exactly $k$ edges between $X'$ and $Y'$, then the assignment $\pi$ as defined above is EFX. 

    Towards the other direction, let $\pi$ be an EFX assignment. We first observe that $|\pi_i\cap (X\cup Y)| = 1$ for all $i\in N$. For the sake of contradiction, assume this is not the case. Since $|X\cup Y|= N$, there is $i\in N$ such that $|\pi_i\cap (X\cup Y)| = 0$ and $j\in N$ such that $|\pi_j\cap (X\cup Y)| \ge 2$. Let $a\in \pi_j\cap (X\cup Y)$ and let $\lambda$ be the allocation obtained from $\pi$ by removing $a$. Then $V_i(\lambda) - V_i(\lambdaswap) = V_i(\pi) - V_i(\piswap) + \dij$, so property 2. of Definition~\ref{def:EFX} is satisfied. Moreover, as $j$ still contains at least one item from $X\cup Y$ in $\lambda$, we have $V_i(\lambdaswap) \ge  V_i(\lambda) + \dij +|Z\cap \pi_j| - |Z\cap \pi_i|$. Since $\dij > |Z|$, it follows that $V_i(\lambdaswap) >  V_i(\lambda)$, a contradiction to $\pi$ being EFX. It follows that $|\pi_i\cap (X\cup Y)| = 1$ for all $i\in N$. 
    
    Now, note that $|Z| = |N|-k$, hence there are at least $k$ agents in $N$ without any item from $Z$. Let us show that there are precisely $k$ agents without any item in $Z$ and all the other agents receive exactly one item from $Z$. For the sake of contradiction, let $j$ be an agent such that $|\pi_j\cap Z| \ge 2$. Since $G$ is cubic, there are precisely four edges that share an endpoint with $e_j$. Because $k > 4$, there is an agent $i$ such that $\pi_i\cap Z = \emptyset$ and $e_i$ and $e_j$ do not share an endpoint. Let $\lambda$ be an assignment obtained from $\pi$ by removing some $z\in \pi_j\cap Z$. Then $V_i(\lambda) - V_i(\lambdaswap) = V_i(\pi) - V_i(\piswap) + 1$, so property 2. of Definition~\ref{def:EFX} is satisfied. However, $V_i(\lambdaswap) = V_i(\lambda) + |\pi_j\cap Z\setminus\{z\}|\ge V_i(\lambda)+1$, so $\pi$ is not EFX. It follows that there are precisely $k$ agents without any item from $Z$, while any other agent receives exactly one item from $Z$.

    Consider the subset of agents $I = \{i\in N \mid Y\cap \pi_i \neq \emptyset\}$. We show now that if $i\in I$ and $j\in N\setminus I$ are such that $e_i$ and $e_j$ share an endpoint, then $\pi_j = \{x\}$ for some $x\in X$. Since $j\in N\setminus I$, it holds that $\pi_j\cap (X\cup Y) = \{x\}$ for some $x\in X$, hence we only need to show that $\pi_j\cap Z = \emptyset$. Otherwise, let $\pi_j\cap Z = \{z\}$ and let $\lambda$ be the assignment obtained from $\pi$ by removing $z$. Then $V_i(\lambda) - V_i(\lambdaswap) = V_i(\pi) - V_i(\piswap) + 1$, so property 2. of Definition~\ref{def:EFX} is satisfied: the envy of $i$ towards $j$ is decreased. However, it is not eliminated since $V_i(\lambdaswap) = V_i(\lambda) -\yii + \xii> V_i(\lambda)$, which contradicts to $\pi$ being EFX.

    Let $Y'$ be the set of vertices of $G$ that are endpoints of edges in $E_I=\{e_i\mid i\in I\}$, and let $F$ be the set of edges of $G$ with precisely one endpoint in $Y'$. Since $G$ is cubic and $|E_I|=|I|=|Y|=\frac{3n-k}{2}$, the sum of degrees of vertices in $Y'$ is $3|Y'| = 2 |E_I| + |F| = 3n-k+|F|$. Every edge from $F$ shares an endpoint with an edge from $E_I$ and so any agent $j$ such that $e_j\in F$ does not have any item from $Z$. In particular, $|F| \le k$ and hence $|Y'| =n - \frac{k - |F|}{3} \le n$. 
    
    Let $Y''$ be any subset of $V(G)\setminus Y'$ of size $\frac{k-|F|}{3}$. Consider the partition $(Y'\cup Y'', V(G)\setminus (Y'\cup Y'')$ of $V(G)$ into two sets of size $n$. Since $G$ is cubic, there are at most $|F| + 3\frac{k-|F|}{3}=k$ edges with precisely one endpoint in each part. Therefore, $(G,k)$ is YES-instance of \textsc{Min Bisection}.
\end{proof}

Our negative results strongly indicate that in order to derive some positive results for general valuations, we {\em have to} further restrict ourselves. 
Interestingly, we show that if we combine the two parameters for which the problem is intractable, when we consider them independently, the problem becomes fixed-parameter tractable. 
Hence, we provide a dichotomy with respect to this combination of parameters.
Observe that Theorems~\ref{thm:NP-h:three-agents} and~\ref{thm:EFX:NPh:d} actually show para-\NP-hardness for the problem when parameterized by agents, item types, and number of different values in the valuation functions.

We start our journey for the first algorithmic results with instances, where the number of agents and the number of item types is bounded. 

\begin{theorem}\label{thm:FDE:EFX:FPT:types}
    The \FDwE{$\phi$} problem, where $\phi\in\{\text{EF},\text{EF1},\text{EFX}\}$, is fixed-parameter tractable when parameterized by the number of different item types $\types$ and the number of agents $|\agents|$ combined.
\end{theorem}
\begin{proof}
    As the first step of our algorithm, we partition the items according to their types $T = \{T_1,\ldots,T_\types\}$ and compute the size $n_{T_1},\ldots,n_{T_\types}$ of each partition. For the rest of this proof, we will use the notion of \emph{bundle-types}. The bundle-type is defined by the subsets of different item types that has at least one representative in the bundle. It is easy to see that there are at most $2^\Oh{\types}$ bundle-types in total.

    Now, we guess for each agent its bundle-type. There are~$2^{{\types}^\Oh{|\agents|}}$ such guesses and for each guess, we construct an ILP that verifies whether the guess satisfies the given notion of fairness; in what follows, we assume envy-freeness (EF), but later we show how to tweak the construction to handle also the other notions. We denote by $B(i)$ the set of item types present in the agent's $i\in\agents$ bundle according to our guess. 
    In addition,
    we extend the definition of valuation to types and use $V_i(j,t)$ to denote how agent $i\in\agents$ values the item of type $t\in T$ assigned to agent $j\in\agents$.

    Our ILP contains $\Oh{|\agents|\cdot\types}$ variables $x_{t,i}$ representing the number of items of type $t$ assigned to the agent $i$. The constraints are as follows.
    \begin{align}
        \forall t\in T\colon                      & \sum_{i\in\agents} x_{t,i} = n_t \\
        \forall i\in\agents\colon  \forall t\in T\colon & x_{t,i} \begin{cases}
            \geq 1 & \text{if }t\in B(i),\\
             = 0 & \text{if }t\not\in B(i)
        \end{cases}\\
        \forall i,j\in\agents\colon               & \sum_{t\in T} \left(x_{t,i}\cdot V_i(i,t) + x_{t,j}\cdot V_i(j,t)\right)\nonumber\\
                                            &\geq \sum_{t\in T} \left(x_{t,j}\cdot V_i(i,t) + x_{t,i}\cdot V_i(j,t)\right) \label{eq:EF}%
    \end{align}

    The first set of $\Oh{\types}$ constraints ensures that all items are allocated. The second set of constraints of size $\Oh{|\agents|\cdot\types}$ secures that the items allocated to each agent correspond to the guessed bundle. Finally, the third set of $\Oh{|\agents|^2}$ constraints is to verify that the outcome is envy-free.
    
    If we are interested in EF1 allocations, then just before we construct the verification ILP, we compute for each pair of distinct agents $i,j\in N$ the item-type that decreases the envy the most. Observe that this depends solely on item-types present in their bundles and can be computed in polynomial time. Finally, we incorporate this removal into the ILP. Let $i,j\in N$ be two distinct agents and let $t\in T$ be a type of item that decreases envy from $i$ towards $j$ the most. Then we tweak the constraint \eqref{eq:EF} for this pair of agents to reflect the removal of one item of type $t$.

    For EFX allocations, the idea is very similar to the case of EF1 allocations. However, this time, instead of precomputing an item-type that decreases the envy the most, we determine an item-type that decreases the envy the least, but still by some positive value.
    
    It is well known that ILPs with parameter-many variables can be solved by a fixed-parameter algorithm~\citep{Lenstra1983,Kannan1987,FrankT1987}, and the theorem follows.
\end{proof}

Theorem~\ref{thm:FDE:EFX:FPT:types} above can be used to almost immediately give us the following corollary. The key ingredient here is to show that whenever the number of different values and the number of agents is bounded, so is the number of different item-types.

\begin{corollary}
    The \FDwE{$\phi$} problem, where $\phi\in\{\text{EF},\text{EF1},\text{EFX}\}$, is fixed-parameter tractable when parameterized by the number of agents $|N|$ and the number of different values~$d$ in agents' preferences.
\end{corollary}
\begin{proof}
    Item types are defined as the number of different $|N|\times|N|$ matrices formed from agent preferences. If the number of agents and the number of different values is bounded, we have $\types \in d^\Oh{|N|^2}$ and, therefore, we can use \Cref{thm:FDE:EFX:FPT:types} to decide the given instance in \FPT time.
\end{proof}

Our next algorithmic result shows that if the number of items is constant, we can solve the problem efficiently. Again, our algorithm is more general and can handle all the fairness notions we study in this work.

\begin{proposition}
    The \FDwE{$\phi$} problem, where $\phi\in\{\text{EF},\text{EF1},\text{EFX}\}$ is in \XP when parameterized by the number of items $|\items|$.
\end{proposition}
\begin{proof}
    Our algorithm is a simple brute-force that, for each item $a\in\items$, guess the agent the item is allocated to in a fair allocation and it verifies in polynomial, whether the guess is correct. If the guessed allocation is fair, we return \emph{yes}. Otherwise, we continue with another possibility. Overall, there are $|\agents|^\Oh{|\items|}$ possible allocations and each fairness notion can be checked in polynomial time. That is, the \FDwE{$\phi$} problem, $\phi\in\{\text{EF},\text{EF1},\text{EFX}\}$, is clearly in \XP when parameterized by the number of items. As we exhaustively tried all possible allocations, the algorithm is trivially correct.
\end{proof}

The following result shows that we cannot hope for an \FPT algorithm for \FDwE{EF} even in severely restricted instances. This contrasts with the setting without externalities, where an \FPT algorithm for the problem trivially exists because of the simple fact that once the number of items is smaller than the number of agents, the EF solution can never exist.

\begin{theorem}
    The \FDwE{EF} is \Wh when parameterized by the number of items $|\items|$, even under binary valuations and with no chores. Unless ETH fails, there is no algorithm solving such instances in $g(|\items|)\cdot (|\agents|)^{o(\sqrt{|\items|})}$ time for any computable function $g$.
\end{theorem}

\begin{proof}
    We prove the hardness by giving a parameterized reduction from the \textsc{Multicolored Clique} problem. Here, we are given an integer $k$, a $k$-partite graph $G=(U = U_1\cup U_2 \cup \cdots \cup U_k, E)$, and our goal is to find a set $K\subseteq U$ of size $k$ such that $G[U]$ is a complete graph. The \textsc{Multicolored Clique} problem is known to be \Wh when parameterized by $k$, even if all parts $U_i$ are of the same size~$n'$ and there is the same number $m'$ of edges between each pair of distinct parts $U_i$ and $U_j$~\citep{Pietrzak2003,FellowsHRV2009}. For each $i\in[k]$, we call the set $U_i = \{u_i^1,\ldots,u_i^{n'}\}$ a \emph{color class} and for two distinct $i,j\in[k]$, $i<j$, we denote by $E_{i,j} = \{e_{i,j}^1,\ldots,e_{i,j}^{m'}\}$ the set of all edges between $U_i$ and $U_j$; formally, $E_{i,j} = \{\{u_i^\ell,u_j^{\ell'}\}\in E\mid u_i^\ell \in U_i \land u_j^{\ell'} \in U_j\}$. We can also assume that $m' > n' > k^2 + k$ as otherwise, the size of the instance is bounded in terms of the parameter, and we can solve it by a trivial brute-force algorithm in \FPT time.

    \medskip
    \noindent\textbf{Construction.}\hspace{0.25cm}
    Given an instance $\mathcal{I} = (G,k)$ of the \textsc{Multicolored Clique} problem, we construct an equivalent instance $\mathcal{J}$ of the \FDwE{EF} problem as follows. 
    
    First, we define the set of agents $\agents$ and the set of items $\items$. For every color class $U_i$, $i\in[k]$, we create $n'$ \emph{vertex-agents} $w_{i}^1,\ldots,w_{i}^{n'}$, and we denote the set of all vertex-agents for a single color $i$ by~$W_i$. We say that vertex-agent $w_i^\ell$, $\ell\in[n']$, is of color $i$. The vertex-agents of color $i$ are in one-to-one correspondence with the vertices of color $i$ in the input graph $G$. Similarly, for each pair of distinct colors $i,j\in[k]$, $i<j$, we create $m'$ \emph{edge-agents} $f_{i,j}^1,\ldots,f_{i,j}^{m'}$ which are in one-to-one correspondence with edges of $E_{i,j}$. By $F_{i,j}$, we denote the set of all edge-agents corresponding to edges in $E_{i,j}$. Sometimes, we say that agents are from the same \emph{group} if all belong to the same $W_i$ or $F_{i,j}$. In total, there are $|U| + |E|$ agents. The set of items consists of $k$ \emph{selection-items} $a_1,\ldots,a_k$ and of a single \emph{incidence-item}~$a_{i,j}$ for each pair of distinct $i,j\in[k]$, $i < j$. It holds that $|\items| = k + \binom{k}{2} = \Oh{k^2}$. The high-level idea behind the construction is that, by allocation of selection-items, we select a vertex for the clique in~$G$. The incidence-items then ensure that in the original graph, there is an edge between the vertices selected by selection-items.
    
    To secure the desired behavior of the equivalent instance $\mathcal{J}$, we define the valuations of the agents as follows.
    \begin{description}
        \item[Edge-agent $f_{i,j}^\ell$.]
            The edge-agents do not care where the selection items are allocated, as the value they gain from these items is always zero. Formally, for every selection-item~$a$ and for each agent $x\in\agents$, we set $V_{f_{i,j}^\ell}( x, a ) = 0.$
            
            The situation with incidence-items is a little bit more complicated. Let $a_{i',j'}$ be an incidence-item. The edge-agent $f_{i,j}^\ell$ gains utility $1$ if and only if $i' = i$, $j' = j$, and the item $a_{i',j'}$ is allocated to some edge-agent from the group $F_{i,j}$. Formally, the valuations for an incidence-item $a_{i',j'}$ are as follows:
            \[
                V_{f_{i,j}^\ell}(f,a_{i',j'}) = \begin{cases}
                    1 & \text{if } f \in F_{i,j} \text{, } i' = i\text{, and } j' = j,\\
                    0 & \text{otherwise.}
                \end{cases}
            \]
            Observe that neither selection-items nor incidence-items are chores for edge-agents.
        \item[Vertex-agent $w_i^\ell$.] 
            The vertex-agent $w_i^\ell$ has positive utility if the corresponding selection-item $a_i$ is allocated to any vertex-agent of color $i$. For all other selection-items $a_{i'} \not= a_i$, the utility of agent $w_i^\ell$ is zero regardless of to whom $a_{i'}$ is allocated. This ensures that the selection-items are allocated to vertex-agents of the correct color~$i$. Formally, the valuation of the agent $w_i^\ell$ for a selection-item $a_{i'}$ is as follows:
            \[
                V_{w_i^\ell}(w,a_{i'}) = \begin{cases}
                    1 & \text{if } w \in W_i \text{ and } i' = i,\\
                    0 & \text{otherwise.}
                \end{cases}
            \]
            
            Next, we define the value the vertex-agent $w_i^\ell$ gains from an incidence-item $a_{i',j}$. As with selection-items, the vertex-agent has zero value for all incidence-items with $i'\not=i$ and $j\not=i$. In other words, the `interesting' selection-items for $w_i^\ell$ are only those that represent an edge with one endpoint in the color class $U_i$. Specifically, the vertex-agent $w_i^\ell$ has a value $1$ if and only if an interesting selection-item is allocated to $w_i^\ell$ itself or to an edge-agent representing an edge $e \in E$ such that $u_i^\ell \not\in e$. It is worth observing that $w_i^\ell$ has zero value if a selection-item $a_{i,j}$ or $a_{i',i}$ is allocated to an agent representing an edge $e$ such that $u_i^j \in e$. Formally, the valuation of the vertex-agent $w_i^\ell$ for an incidence-item $a_{i',j}$ is as follows:
            \[
                V_{w_i^\ell}(w,a_{i',j}) = \begin{cases}
                    0 & \text{if } i' \not= i \text{ and } j \not = i,\\
                    1 & \text{if } w = w_i^\ell,\\
                    1 & \text{if } w = f_{i',j}^{\ell'} \text{ and } u_i^\ell \not\in e_{i',j}^{\ell'}\text{, and}\\
                    0 & \text{if } w = f_{i',j}^{\ell'} \text{ and } u_i^j \in e_{i',j}^{\ell'}.
                \end{cases}
            \]
            It is again easy to see that no item is a chore for vertex-agents.
    \end{description}
    
    This finishes the description of the construction. Note that the valuation functions are clearly binary and we have already discussed that no item is a chore.

    \medskip
    \noindent\textbf{Left to Right Implication.}\hspace{0.25cm}
    For the correctness, assume first that $\mathcal{I}$ is a \emph{yes}-instance and $K = \{u_{1}^{\ell_i},\ldots,u_{k}^{\ell_k}\}$ is a solution clique. We create an allocation $\pi$ in the following way. For each selection-item $a_i$, $i\in[k]$, we set $\pi(a_i) = w_i^{\ell_i}$. Next, for each incidence-item $a_{i,j}$, we set $\pi(a_{i,j}) = f_{i,j}^{\ell}$, where $e_{i,j}^{\ell} = \{u_i^{\ell_i}, u_j^{\ell_j} \}$. Observe that such an edge-agent $f_{i,j}^{\ell}$ always exists, since $K$ is a clique in $G$ and we created an edge-agent for each edge of $G$. We claim that $\pi$ is an envy-free allocation.
    
    Clearly, for all $f_{i,j}^\ell$, $\ell\in[m']$, the utility is exactly $1$ and, by the definition of their valuations, they are worse in every allocation in which they swap bundles with an agent from other groups, and their utility remains the same if they swap with an edge-agent from the same group. 
    
    By the same argumentation, the vertex-agents do not envy other vertex-agents. Therefore, what remains to check is that no vertex-agent $w_i^\ell$ envy to an edge-agent $f_{i',j}^{\ell'}$. If $i\not\in \{i',j\}$, then there is no envy from $w_i^\ell$ as his utility for the item allocated to $f_{i',j}^{\ell'}$ is always zero and there are no chores. To finalize the implication, we examine separately two cases based on the size of the $w_i^\ell$'s bundle.
    
    Let $w_{i}^\ell$ be a vertex-agent with an empty bundle. As there cannot be envy between two agents with empty bundles, we can assume that the bundle of $f_{i',j}^{\ell'}$ is not empty. Moreover, as in the previous argument, we have $i\in\{i',j\}$ and, without loss of generality, assume that $i'= i$. Assume that $u_i^\ell\not\in e_{i',j}^{\ell'}$. Then, we have $V_{w_i^\ell}(w_i^\ell,a_{i',j}) = 1$ and $V_{w_i^\ell}(f_{i',j}^{\ell'},a_{i',j}) = 1$. That is, the agent is indifferent whether $a_{i',j}$ is allocated to her or to $f_{i',j}^{\ell'}$, resulting in the agent $w_i^\ell$ not being envious of any other agent. The situation where $u_i^\ell \in e_{i',j}^{\ell'}$ cannot occur in this case by the definition of $\pi$.

    The last remaining case is when the vertex-agent $w_i^\ell$'s bundle is non-empty. As there are no chores, there is never an envy from agents with non-empty bundle towards agents with empty bundles. Again, we can, without loss of generality, assume that $i' = i$. According to the definition of valuations, the agent has zero externality from the item $a_{i,j}$ allocated to $f_{i,j}^{\ell'}$ and would prefer to swap with this agent as $w_i^\ell$ values $a_{i,j}$ $1$ if it is allocated to her. However, by swapping bundles, agent $w_i^\ell$ loses the selection-item allocated to her. Consequently, the utility in both allocations is the same and there is no envy from agent $w_i^\ell$ towards any other edge-agent with non-empty bundle.

    We examined all possible combinations of agents and we proved that $\pi$ is indeed an envy-free allocation. Therefore, $\mathcal{J}$ is also a \emph{yes}-instance.

    \medskip
    \noindent\textbf{Right to Left Implication.}\hspace{0.25cm}
    In the opposite direction, let $\mathcal{J}$ be a \emph{yes}-instance and $\pi$ be an envy-free allocation. To prove that $\mathcal{I}$ is also a \emph{yes}-instance, we use several auxiliary claims. First, we prove that all selection-items are allocated to vertex-agents and incidence-items are allocated to edge-agents.

    \begin{claim}\label{cl:selection-color:incidence-edge}
        Let $a$ be a selection-item and $b$ be an incidence-item. In every-free allocation $\pi$, $\pi(a)$ is a vertex-agent and $\pi(b)$ is an edge-agent.
    \end{claim}
    \begin{claimproof}
        Let us start with incidence-items. Assume that $\pi$ is an envy-free allocation and let $b$ be an incidence-item such that $\pi$ allocates $b$ to a vertex-agent $w$. Let $f$ be an edge-agent such that $V_f(f,b) = 1$. It follows by the definition of valuations that this agent has nonzero utility if and only if the item $b$ is allocated to some agent in the $f$'s group. Consequently, the current utility of the edge-agent $f$ is zero and can be increased by swapping his bundle with the bundle of agent $w$. Thus, there is an envy of agent $f$ towards agents $w$, which contradicts the fact that $\pi$ is an envy-free allocation. Consequently, no incidence-item is allocated to a vertex-agent.
    
        For the second part of the claim, suppose that the selection-item $a$ is allocated to an edge-agent $f_{i,j}^\ell$. Let $w$ be, without loss of generality, a vertex-agent of color $i$ such that $\pi_w = \emptyset$. Observe that such an agent always exists, since $n' > k^2 + k > |\items|$. Then, the utility of $w$ is $V_w(\pi) = \mu + V_w(\pi_w) = \mu + 0 = \mu$ in $\pi$ and $V_w(\pi^{w \leftrightarrow f_{i,j}^\ell}) = \mu + V_w(\pi_{f_{i,j}^\ell})$ in the allocation where $w$ and $f_{i,j}^\ell$ swap their bundles. Since $\pi_{f_{i,j}^\ell}$ contains the item $a$ and maybe some incidence-items the agent $w$ values at least the same if they are allocated to him or to $f_{i,j}^\ell$, the utility of $a$ increases in the latter allocation by at least one. Therefore, there is an envy of $w$ towards $f_{i,j}^\ell$, which contradicts that $\pi$ is an envy-free allocation. Thus, in every envy-free allocation, the selection-items are necessarily allocated to vertex-agents.
    \end{claimproof}
    
    With \Cref{cl:selection-color:incidence-edge} in hand, we know that selection-items are necessarily allocated to vertex-agents and incidence-items are allocated to edge-agents. In the following claim, we even strengthen this property and show that each selection-item is allocated to a vertex-agent of the matching color.

    \begin{claim}\label{cl:selectMatchesColor}
        Let $a_i$, $i\in[k]$, be a selection-item. In every envy-free allocation $\pi$, $\pi(a_i) = w_i^\ell$ for some $\ell\in[n']$.
    \end{claim}
    \begin{claimproof}
        By \Cref{cl:selection-color:incidence-edge}, we know that $a$ is necessarily allocated to some vertex-agent, as otherwise, the allocation $\pi$ would not be envy-free. For the sake of contradiction, let $\pi$ be an envy-free allocation such that $\pi(a_i) = w$, where $w$ is a vertex-agent of color $j\not= i$. Using \Cref{cl:selection-color:incidence-edge}, we know that $w$'s bundle consists only of selection-items. Let $w_i^\ell$ be a vertex-agent of color $i$ with an empty bundle. The utility of $w_i^\ell$ can be decomposed as $\mu + V_{w_i^\ell}(\pi_{w_i^\ell})$, where $V_{w_i^\ell}(\pi_{w_i^\ell})$ is always zero and $\mu$ is the utility agent $w_i^\ell$ has from externalities. Recall that the only positive utility a vertex-agent has from a selection-item is when a selection-item of the same color is allocated to him or any other agent of this color. Consequently, only the externalities from edge-agents contribute to $\mu$, and more importantly, the agent $w_i^\ell$ has zero externality from the bundle of agent $w$. Therefore, if agents $w_i^\ell$ and $w$ swap their bundles, the utility of $w_i^\ell$ will be $\mu + V_{w_i^\ell}(\pi_{w})$, where $V_{w_i^\ell}(\pi_{w})$ is exactly one due to the item $a_i$. That is, agent $w_i^\ell$ has envy towards agent $w$, which is a contradiction to $\pi$ being an envy-free allocation. Thus, all envy-free allocations have to allocate selection-items to the vertex-agents of the matching color.
    \end{claimproof}

    Next, we show that the incidence-items posses similar properties as selection-items, that is, they are allocated to an edge-agent in the correct group.
    \begin{claim}
        Let $a_{i,j}$, $i,j\in[k]$ and $i < j$, be an incidence-item. In every envy-free allocation $\pi$, $\pi(a_{i,j}) = f_{i,j}^\ell$ for some $\ell\in[m']$.
    \end{claim}
    \begin{claimproof}
        By \Cref{cl:selection-color:incidence-edge}, we know that $a_{i,j}$ is allocated to some edge-agent $f_{i',j'}^{\ell'}$. For the sake of contradiction, assume that $\{i',j'\} \cap \{i,j\} \not= \{i,j\}$. Let $f_{i,j}^\ell$ be an edge-agent with an empty bundle in $\pi$. As the item $a_{i,j}$ is allocated to an edge-agent from a different group, the utility of $f_{i,j}^\ell$ is zero. However, by swapping the bundle with $f_{i',j'}^{\ell'}$, his or her utility becomes $1$. Consequently, there is envy, which contradicts that $\pi$ is envy-free. Thus, each envy-free allocation allocates each incidence-item $a_{i,j}$ to an edge-agent $f_{i,j}^\ell$ for some $\ell\in[m']$.
    \end{claimproof}

    Finally, we create a set $K = \{ u_i^\ell \in U \mid \pi(a_i) = w_i^\ell \}$ and claim that it induces a clique of the correct size in $G$. Clearly, the set $K$ is of size $k$ as there is exactly one selection-item for each $i\in[k]$, and by \Cref{cl:selectMatchesColor}, $K$ contains exactly one vertex of each color class. Suppose that there exists a pair of distinct vertices $u_i^\ell,u_{j}^{\ell'} \in K$ such that $\{u_i^\ell,u_{j}^{\ell'}\}\not\in E$. Then, in the equivalent instance~$\mathcal{J}$, the envy-free allocation $\pi$ allocates the selection-item $a_i$ to $w_i^\ell$, the selection-item $a_{j}$ to $w_{j}^{\ell'}$, and the incidence-item $a_{i,j}$ to some~$f_{i,j}^{\ell''}$. Moreover, since there is no edge between $u_i^\ell$ and $u_j^{\ell'}$ in $G$, it follows that $u_i^\ell\not\in e_{i,j}^{\ell''}$ or $u_j^{\ell'} \not\in e_{i,j}^{\ell''}$. Without loss of generality, let $u_i^\ell \not\in e_{i,j}^{\ell''}$. Then there exists a vertex-agent $w_i^*\in W_i$ such that $u_i^* \in e_{i,j}^{\ell''}$ and with an empty bundle. The utility of $w_i^*\in W_i$ is $V_{w_i^*}(\pi) = \mu + V_{w_i^*}(\emptyset) + V_{w_i^*}(f_{i,j}^{\ell''}) = \mu + 0 + 0$, where $\mu$ is the value that agent $w_i^*$ gains from the bundles of all agents except his own and the one of $f_{i,j}^{\ell''}$ (this value remains the same after the swap). After swapping the bundle with the agent $f_{i,j}^{\ell''}$, the agent $w_i^*$ receives the item $a_{i,j}$ that he values $1$, and his overall utility increases. This contradicts that $\pi$ is an envy-free allocation. Therefore, $\{u_i^\ell,u_j^{\ell'}\}\in E$ for every pair of distinct $u_i^\ell,u_j^{\ell'}\in K$, $K$ is a clique, and $\mathcal{I}$ is also a \emph{yes}-instance. This finishes the correctnes of the construction.

    \medskip
    \noindent\textbf{Summary and ETH Lower-Bound.}\hspace{0.25cm}
    It is easy to see that the construction can be performed in polynomial time, uses binary valuations, and all items are goods. We have also shown that $|\items| = \Oh{k^2}$, that is, the reduction is indeed a parameterized reduction. It is well known that \textsc{Multicolored Clique} cannot be solved in time $h(k)\cdot n^{o(k)}$, unless ETH is false~\citep{ChenCFHJKX05}. Therefore, an algorithm running in time $g(|\items|)\cdot n^{o(\sqrt{|\items|})}$ would yield to an algorithm running in time $g(k^2)\cdot n^{o(\sqrt{k^2})} = h(k)\cdot n^{o(k)}$ for \textsc{Multicolored Clique}, which is unlikely. This finishes the proof.
\end{proof}

On the positive side, in the next result, we show that if we relax the notion of fairness a bit, we obtain an efficient algorithm.

\begin{proposition}
    The \FDwE{EF2} problem is fixed-parameter tractable when parameterized by the number of items $|\items|$.    
\end{proposition}
\begin{proof}
    Suppose that $|\agents| \leq |\items|$. Then we can check all possible partitions of items in $|\agents|^\Oh{|\items|}$ time, which is clearly in \FPT. Therefore, let $|\agents| > |\items|$ and $\pi$ be an allocation such that $\forall a_i\in\agents\colon \pi(a_i) = i$. By the definition of the allocation, each bundle is of size at most one. As envy-freeness-up-to-two-items is based on a pairwise comparison of agents, each comparison involves at most $2$ items -- both of them can be removed to eliminate all envy between two agents. Thus, $\pi$ is EF2, finishing the proof.
\end{proof}

We conclude this section with a property that could be of independent interest. Specifically, we show that there are instances where no allocation maximizing Nash social welfare is EFX.
This contrasts the setting of fair division without externalities and additive preferences, as shown by \citet{AmanatidisBFHV2021}.

\begin{proposition}\label{prop:NashSW}
    Let $\ical$ be an instance of Fair Division with Externalities and $\pi^*$ be an allocation maximizing Nash social welfare. Then $\pi^*$ is not necessarily EFX.
\end{proposition}
\begin{proof}
    Let $\ical$ be an instance with $\items=\{a_1,a_2,a_3,a_4\}$, $|\agents|\geq 3$, the valuation function for agent $i\in\{2,\ldots,|N|\}$ and every item $a\in \items$ be defined as
    \[
        V_i(j,a) = \begin{cases}
            1 & \text{if } i = j \text{ or } j = 2,\\
            0 & \text{otherwise,}
        \end{cases}
    \]
    and for every item $a\in \items$ let $V_1(1,a) = 1$ and $0$ otherwise.
    
    First, observe that the items are, from the perspective of the agents, indistinguishable, and the only thing that matters for utilities is the number of items allocated to each agent. Therefore, in the rest of this proof, we assume two allocations with the same number of items allocated to the same agents as equivalent. Based on this, we can easily compute the Nash social welfare for an allocation $\pi$ using the function
    \[
        \operatorname{NW}(\pi) = |\pi_1| \cdot |\pi_2| \cdot \prod_{i\in\{3,\ldots,|N|\}} (|\pi_2| + |\pi_i|),
    \]
    subject to the two constraints: $\sum_{i\in N} \pi_i = 4$ and ${\forall i\in N\colon \pi_i \in \{0,1,2,3,4\}}$. This clearly attains its global maximum in $|\pi_1| = 1$ and $|\pi_2| = 3$. It can be shown either analytically or by a simple argumentation: there has to be at least one item allocated to agent $1$ and at least one item allocated to agent $2$. Moreover, agents $3\ldots N$ do not care whether an item is allocated directly to them or to the agent~$2$, while agent $2$ only benefits from items allocated to him. Therefore, the only allocation $\pi^*$ maximizing Nash social welfare gives one item to agent $1$ and three items to agent $2$.
    
    On the other hand, this allocation is clearly not EFX as the agent $1$ benefits from swapping bundle with agent $2$ even if we remove any item from $\pi^*_2$ -- as stated before, the items are anyway indistinguishable.
\end{proof}

It deserves to highlight that the proof of \Cref{prop:NashSW} uses an instance with only $3$ agents, $4$ items which are all of the same item-type, and binary valuations.

\section{Binary Valuations}
\label{sec:binary}
In this section, we study the special case of $\FDwE{}$ when all the valuations are binary, i.e., all $V_i$ have domain $\{0,1\}$. We would like to stress here that, in the setting without externalities, the domain for binary valuations is usually $\{-1,0,1\}$; $\{0,1\}$ for goods and $\{-1,0\}$ for chores~\citep{AzizLRS2023}. However, in the presence of externalities, there can exist chores even without negative values, so \citet{efx-externalities} defined binary valuations only using the domain $\{0,1\}$.

In our first result, we show that, in fact, binary valuations allow us to capture any scenario where every agent has at most two different values.

\begin{proposition}
    Let $\ical=(N,\items,V)$ be an instance of $\FDwE{\phi}$, where $\phi \in \{ \text{EF, EF1, EFX}\}$. Assume that for every agent $i$, there exist two numbers $x_i$ and $y_i$ such that $V_i(j,a) \in \{x_i,y_i\}$ for every agent $j$ and item $a$. Then $\ical$ can be transformed in linear time in the equivalent instance of $\FDwE{\phi}$ with the same sets of agents and objects and binary valuations.
\end{proposition}
\begin{proof}
    Without loss of generality, we assume that for every agent $i$, $x_i < y_i$. For each pair of agents $i$ and $j$, we denote by $\items_{i\to j}$ the subset of items from which the agent $i$ receives a strictly bigger value once they are allocated to $j$ instead of~$i$, i.e., $\items_{i\to j}=\{a\in \items: V_i(i,a)=x_i < y_i= V_i(j,a)\}$. Similarly, let $\items_{i\to j}$ be the subset of items from which $i$ receives a strictly bigger value once they are allocated to $i$ itself: $\items_{i\leftarrow j}=\{a\in \items: V_i(i,a)=y_i > x_i= V_i(j,a)\}$. We construct the new valuations $V_i'$ of $i$ as follows:
    \[
        V'_i(j,a) = \begin{cases}
            0 & \text{if  } \:V_i(j,a)=x_i,\\
            1 & \text{if  } \:V_i(j,a)=y_i,
        \end{cases}
    \] 
    and define the new instance $\ical'=(N,\items,V')$. To see that~$\ical$ and $\ical'$ are equivalent, consider arbitrary allocation $\pi$ of items in $\items$ to agents in $N$ and fix two agents $i$ and $j$. The difference between values of $i$ after and before the swap with $j$ in $\ical$ is $V_i(\piswap)-V_i(\pi)=$
    \begin{multline*}
   \sum_{a\in \pi_j} (V_i(i,a)-V_i(j,a)) + \sum_{a\in \pi_i} (V_i(j,a)-V_i(i,a))=\\ 
    \sum_{a \in \pi_j \cap \items_{i \leftarrow j}} (y_i-x_i) + 
    \sum_{a \in\pi_j \cap \items_{i \to j}} (x_i-y_i)+\\ 
    \sum_{a\in \pi_i\cap \items_{i \leftarrow j}} (x_i-y_i)+
    \sum_{a\in \pi_i\cap \items_{i \to j}} (y_i-x_i)
    =(y_i-x_i)(m_1 -m_2)
    \end{multline*}
    where $m_1$ and $m_2$ are precisely the numbers of items removal of which decreases (increases correspondingly) the envy of $i$ towards $j$. In other words:
    $$m_1=|(\pi_j \cap \items_{i \shortleftarrow j})\cup (\pi_i \cap \items_{i \shortrightarrow j})|$$
    $$m_2=|(\pi_i \cap \items_{i \shortleftarrow j})\cup (\pi_j \cap \items_{i \shortrightarrow j})|$$
    Analagously, $V'_i(\piswap)-V'_i(\pi)=m_1-m_2$, therefore
    $V_i(\piswap)-V_i(\pi)=(y_i-x_i)(V'_i(\piswap)-V'_i(\pi))$
    and similarly
    for any allocation $\lambda$ obtained from $\pi$ by removing some item. Hence, $\pi$ is EF (EF1) with respect to the valuation~$V'$ if and only if it is EF (EF1) with respect to the valuation~$V$. For the EFX equivalence, observe that removal of any item~$a$ decreases the envy of $i$ towards $j$ in $\ical$ if and only if it decreases the envy in $\ical'$: in this case, the decreases are precisely $y_i-x_i$ and 1 correspondingly.
\end{proof}

Next, we show that the notions of EFX and EF1 allocations coincide for binary valuations.

\begin{proposition}
     \label{bin_EF1_EFX}
     Let $\ical$ be an instance of Fair Division with Externalities with binary valuations and let~$\pi$ be some allocation of items. Then $\pi$ is EFX if and only if $\pi$ is EF1.
\end{proposition}
\begin{proof}
    Observe that every EFX allocation is EF1 by definition, even for general valuations. For another direction, assume that $\pi$ is EF1 allocation. If agent $i$ envies $j$, there must be some item $a$ such that removal $a$ from $\pi$ eliminates the envy. Since the valuations are binary, removal of any item can decrease the envy by at most one, from which we conclude that $V_i(\pi) - V_i(\piswap) =1$. But then the removal of any other item, decreasing $V_i(\pi) - V_i(\piswap)$, eliminates the envy as well. Hence, $\pi$ is EFX allocation.
\end{proof}

\citet{AzizLRS2023} showed that for instances with three-agents, no chores, and binary valuations, an EF1 allocation always exists and can be found in polynomial time. Therefore, by \Cref{bin_EF1_EFX}, we obtain the same guarantee also for~EFX.

\begin{theorem}
    Every instance of Fair Division with Externalities with three agents, binary valuations, and no weak-chores, admits an EFX allocation which can be computed in polynomial time.
\end{theorem}

\section{Correlated Valuations}\label{sec:correlated}
One of the special cases of valuations with externalities are so-called \emph{agent-correlated} valuations, where an agent $i\in N$ receives for an item $a\in \items$:
\begin{itemize}
    \item the best value $v_{i,a}$ if $a$ is allocated to $i$ and
    \item some part $(1-\tau_{i,j})v_{i,a}$ of the best value, $\tau_{i,j}>0$, if $a$ is allocated to another agent $j$. 
\end{itemize} 
One can imagine that the coefficient $\tau_{i,j}$ indicates the degree of friendship between $i$ and $j$. %
We extend this model by adding item-correlations, represented by coefficients $\mu_{i,a}$ for each agent $i$ and item $a$, so the valuations have the following form:
$V=\{V_i(j,a):i,j\in N, a\in \items\}$ such that for every pair of agents $i,j\in N$ and item $a\in \items$,
\begin{equation*}
V_i(j,a)=\begin{cases}
      v_{i,a} & \text{ if $i=j$},\\
      (1-\tau_{i,j}\mu_{i,a})\cdot v_{i,a} & \text{ 
 otherwise,} 
    \end{cases}
\end{equation*}
for some $v_{i,a}$ and $\mu_{i,a}$, where $\tau_{i,j}>0$. 
We call such valuations \textit{agent-item-correlated}. Intuitively, $\mu_{i,a}$ shows how important it is for the agent $i$ that they and no one else receives the item $a$. Surprisingly, it turns out that agent-correlated valuations and even their item-correlated generalizations can be reduced to the valuations without externalities and, therefore, we can use classic algorithms and guarantees from the theory of fair division without externalities~\citep{LiptonMMS04,CaragiannisKMPSW19,AzizCIW2021}. In particular, we get that in these settings EF1 allocations always exist and can be found in polynomial time. 

\begin{theorem}
    Let $\phi \in \{\text{EF, EF1, EFX}\}$. Any instance $\mathcal I$ of $\FDwE{\phi}$ with agent-item-correlated valuations can be transformed in linear time into the equivalent instance $\mathcal I'$ of $\FD{\phi}$ (with no externalities) with the same sets of items and agents.
\end{theorem}
\begin{proof}
    Let $\mathcal I=(\agents, \items, V)$, we construct the instance $\mathcal{I'}=(\agents, \items, U)$ of $\FD{\phi}$ with $U_i(a)=\mu_{i,a} v_{i,a}$ for every agent $i\in N$ and every item $a \in A$. To see the equivalence of $\mathcal I$ and $\mathcal I'$, fix arbitrary allocation of items $\pi$. For every pair of agents $i$ and $j$, we have $V_i(\piswap)-V_i(\pi)=
        \sum_{a\in \pi_j} (V_i(i,a)-V_i(j,a)) + \sum_{a\in \pi_i} (V_i(j,a)-V_i(i,a))=\sum_{a\in \pi_j} \tau_{i,j}\mu_{i,a}v_{i,a} - \sum_{a\in \pi_i} \tau_{i,j}\mu_{i,a}v_{i,a}$, or in terms of $U$ it is $V_i(\piswap)-V_i(\pi)=
        \tau_{i,j} (U_i(\piswap)-U_i(\pi)).$
    Since $\tau_{i,j}>0$, $i$ envies $j$ in $\mathcal I$ if and only if $i$ envies~$j$ in $\mathcal I'$. In particular, the allocation $\pi$ is envy-free for $\mathcal I'$ if and only if it is envy-free for $\mathcal I$. The same holds for any allocation obtained from $\pi$ by removing one item. Therefore, $\mathcal I$ and $\mathcal I'$ are equivalent as instances of EF or EF1.
    
    For EFX, it remains to show that whenever $i$ envies $j$, the removal of any item from $\pi$ decreases the envy in $\mathcal{I}$ if and only if it decreases the envy in $\mathcal{I'}$. Let $\lambda$ be the allocation obtained from $\pi$ by removing some item $a_0$ that decreases envy of $i$ towards $j$ in $\mathcal I$, then we have:
    \begin{multline*}
        V_i(\lambdaswap)-V_i(\lambda)< V_i(\piswap)-V_i(\pi)\iff \\\tau_{i,j}(U_i(\lambdaswap)-U_i(\lambda))< \tau_{i,j}(U_i(\piswap)-U_i(\pi))\\ \iff U_i(\lambdaswap)-U_i(\lambda)< \tau_{i,j}U_i(\piswap)-U_i(\pi),
    \end{multline*}
    since $\tau_{i,j}>0$, which concludes the proof.                        
\end{proof}

\subsection{Examples of agent-item-correlated valuations}

For one real-life scenario, imagine that our agents are partitioned into~$t$ teams $T_1,\ldots,T_t$. Let $c<1$ be some constant. The valuation of every agent is defined as follows: for an item $a\in \items$, $V_i(i,a)=v_{i,a}$. Moreover, $V_i(j,a)=c\cdot v_{i,a}$ if $j\ne i$ is a teammate of $i$ and $V_i(j,a)=0$ otherwise. We call such valuations \emph{team-based}. 

Therefore, in team-based valuations, agents always want to receive items themselves, but otherwise they prefer items to be given to their teammates. This situation can be captured by setting $\mu_{i,a}=1$ along with $\tau_{i,j} = 1-c$ if $i$ and $j$ are in the same team and $\tau_{i,j} = 1$ otherwise.

\begin{corollary}
    For any instance $\ical$ of Fair Division with Externalities with team-based valuations, an EF1 allocation always exists and can be found in polynomial time.
\end{corollary}

For another example of agent-item-correlated valuations, assume that the agents form a graph $G=(N,E)$, which we will call a network. For each item $a$ and agent $i$, if $i$ receives $a$, this also contributes to the values of any other agent $j$, and the contribution depends on the distance $d_{i,j}$ between $i$ and $j$ in the network $G$, namely, $V_i(j,a)=(1-d_{i,j}\mu_{i,a})\cdot v_{i,a}$. We call such valuations \emph{network-based}.

\begin{corollary}
    For any instance $\ical$ of Fair Division with Externalities with network-based valuations, an EF1 allocation always exists and can be found in polynomial time.
\end{corollary}

For instance, network-based valuations would capture the following scenario: a fixed amount of new transport stops (which will be items) should be added, and there are few possible locations (corresponding to agents) for them. The goal is to ensure that from every location one can reach stops that are not too far. To bring $\mu_{i,a}$ into play, imagine that some stops are more important to have nearby then the others. For instance, it is not crucial to have a railway connection in the vicinity, but highly recommended to ensure that there are underground stations close enough.

\section{Conclusions}\label{sec:conclusions}

In this work, we continue the line of research on fair division of indivisible items with externalities initiated by \citet{efx-externalities}. In contrast to previous work, we study the problem from the perspective of computational complexity. To this end, we provide strong intractability results for various restrictions of the problem. On the other hand, we provide several fixed-parameter algorithms that, together with previously mentioned hardness, paint a complete complexity picture of fair division with externalities with respect to its natural parameters. Later, we additionally focus on restricted valuations, providing many properties that lead to previously unknown existence guarantees. 

Our algorithmic results leave open a very intriguing question.
What is the complexity of deciding whether an EF allocation exists when $|\agents|>|\items|$ and we parameterize by the number of items? In the absence of externalities the answer to this is trivial; there is no EF allocation! With externalities though, there is a very easy \XP algorithm, but for fixed-parameter tractability the problem becomes very thought provoking.

It is very common in the fair division literature to combine \emph{fairness} of the outcome with its \emph{efficiency}. Arguably, the most widely studied efficiency notions in the context of fair division are Pareto optimality (PO) and social welfare. 
In our last result, we show that these two notions on their own are easy to compute even under the presence of externalities. Therefore, it is natural to ask for a complexity picture of different combinations of fairness and efficiency.

\begin{proposition}
    Given an instance $\mathcal{I}$ of Fair Division with Externalities, there is a polynomial-time algorithm that finds an allocation which is Pareto optimal and maximizes utilitarian social welfare.
\end{proposition}
\begin{proof}
    First, we show how to find an allocation which maximizes social welfare. We iterate over all items, and we allocate this item to an agent so that the increase in overall utility is the highest possible. In the case of ties, we allocate the item arbitrarily. Since the valuations are additive, this algorithm clearly produces allocation maximizing utilitarian social welfare, and the running time is $\Oh{|\items|\cdot|\agents|}$. Now, we show that this allocation is Pareto optimal. Let $\pi$ be an allocation produced by the previous algorithm and assume that it is not PO. Then, there exists an allocation $\pi'$ such that everyone's utility is at least the utility in $\pi$ and at least one agent is strictly better. Therefore, $SW(\pi') > SW(\pi)$, which contradicts $\pi$ being the MSW allocation.
\end{proof}

Nevertheless, the most appealing question for future research is the (non-)existence of EF1 allocations under binary or general valuations, even if we have only three agents. We conjecture that for binary valuations, EF1 allocation always exists -- and, therefore, due to our results, also EFX allocations always exist in this setting. However, despite many attempts, the proof seems highly non-trivial as already with the no-chores assumption, \citet{AzizLRS2023} used many branching rules, ramified case distinction, and even computer program to verify some branches.
For general valuations, we are more skeptical.

\section*{Acknowledgements}

This work was co-funded by the European Union under the project Robotics and advanced industrial production (reg. no. CZ.02.01.01/00/22\_008/0004590).
Argyrios Deligkas acknowledges the support of the EPSRC grant EP/X039862/1.
Viktoriia Korchemna acknowledges the support of the Austrian Science Fund (FWF, project Y1329). 
Šimon Schierreich acknowledges the additional support of the Grant Agency of the Czech Technical University in Prague, grant No. SGS23/205/OHK3/3T/18.

\bibliographystyle{plainnat}
\bibliography{references}

\end{document}